\documentclass[twocolumn]{autart}
\usepackage{graphicx} 
\usepackage{amsmath}
\usepackage{amssymb}
\usepackage{bbm}
\usepackage{wrapfig}
\usepackage{mathtools}

\newcommand{\mhyphen}{\text{-}}

\begin{document}

\begin{frontmatter}

\runtitle{Stabilizing Communication Topology}  

\title{Data-Driven Communication Topology and Distributed Controller Synthesis: Control-Aware and Co-Design\thanksref{footnoteinfo}} 

\thanks[footnoteinfo]{This paper was not presented at any IFAC 
meeting. Corresponding author M.C.A. Nestor. Tel. +447960-001-549.}

\author[Imperial]{Michael C. A. Nestor}\ead{m.nestor22@imperial.ac.uk},    
\author[Tsinghua]{Jiaxin Wang}\ead{jiaxinwangthu@gmail.com},
\author[Imperial]{Fei Teng}\ead{f.teng@imperial.ac.uk}               

\address[Imperial]{Department of Electrical and Electronic Engineering, Imperial College London, Exhibition Road, London, SW7 2AZ, UK}  
\address[Tsinghua]{Department of Electrical Engineering, Tsinghua University, Haidian District, Beijing, 100084, China}


\begin{abstract}                          
In distributed control schemes, the communication topology that defines information sharing between agents must be designed prior to online control execution. Previous data-driven works typically decouple the topology and controller design problems. This paper investigates data-driven approaches to topology design that guarantee the existence of a stabilizing controller, whilst trading off the level of required communication against control performance. We formulate a data-driven \emph{co-design} scheme that directly co-optimizes the topology and controller over the link costs and closed-loop control cost. When the topology designer does not have access to the agents' control objectives, we propose a \emph{control-aware} approach that couples the two design problems within a sequential scheme. In both cases, the topology is optimized within a mixed-integer semidefinite program (MISDP). Simulations including the IEEE 14-bus power system test case demonstrate the effectiveness of both co-design and control-aware schemes. Numerical examples show lower computation time for the control-aware case than for co-design.
\end{abstract}

\end{frontmatter}

\section{Introduction}
\label{sec:intro}




In the distributed control paradigm, each control agent directly exchanges information with a subset of the other agents. The \emph{communication topology} represents the structure defining information flow between agents, and can be considered a design variable.
The choice of communication topology may limit the achievable closed-loop performance or even create an infeasible structured stabilization problem \cite{Pichai-et-al-1984-SFMs}. In this paper, we consider the coupling between communication topology and controller design.
\begin{figure}[h]
   \centering
   \includegraphics[width=0.75\linewidth]{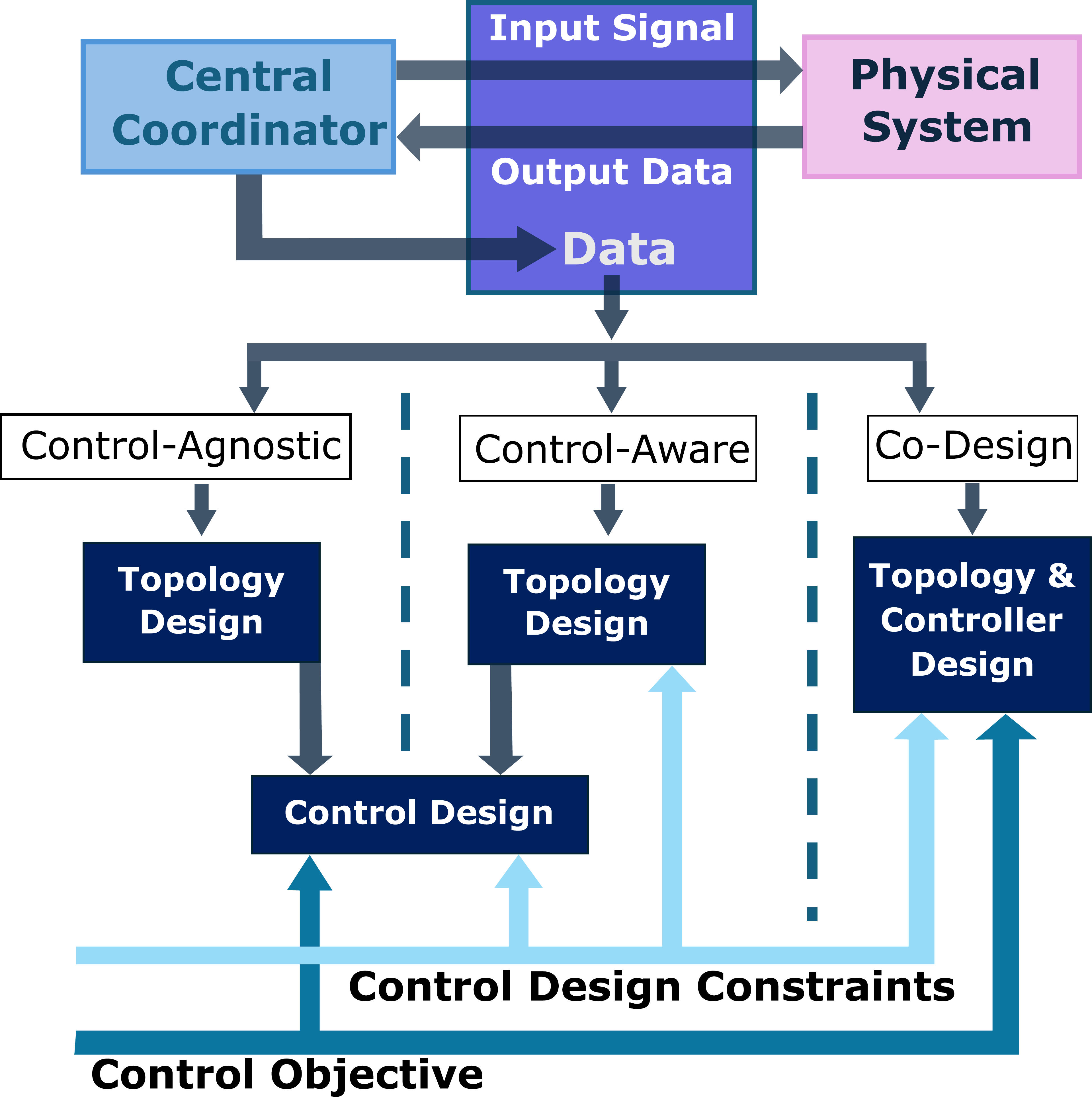}
   \caption{Our framework of control-agnostic, control-aware, and co-design approaches to communication topology design.}
   \label{fig:topology_design_framework}
\end{figure}
Our overall framework is composed of co-design and sequential design classes, as summarized in Fig. \ref{fig:topology_design_framework}. We propose data-driven \emph{co-design}, which selects communication links and controller gains together, and \emph{control-aware} topology design, which selects links using a performance proxy and constraints that preserve feasibility of the subsequent controller synthesis.

\paragraph*{Related Work}
Interconnected linear systems with unstable Decentralized Fixed Modes (DFMs) cannot be stabilized without communication \cite{Wang-Davison-1973-DFMs}, whilst a system with unstable Structural Fixed Modes (SFMs) \cite{Pichai-et-al-1984-SFMs} for a given feedback structure cannot be stabilized by a controller with that structure. 
It is therefore imperative to ensure that the communication topology does not give rise to a structure with unstable SFMs. 
Many works focus solely on designing a topology under which no SFMs exist (\cite{Pequito-2015-Min-Cost-Constrained-I-O-Control-Config,Moothedath-2018-Minimum-Cost-Feedback-Selection-Arbitrary-Pole-Placement-Structured-Systems}) or under which a stabilizing controller exists (\cite{Sturz-2017-Fixed-Mode-Elimination-by-Minimum-Communication,Mozafari-2012-Oscillator-Network-Synchronization-Distributed,Xu-2024-Comms-Topology-Optimization-Time-Delay-CPMG,Mosalli-Babazadeh-2022-Stabilizing-Control-Structures}). These are examples of a modular, sequential approach to communication and control design that first determines a stabilizing control structure, then designs the controller subject to that structure. However, control performance is not considered in the structure design, potentially leading to stabilizing but poorly performing controllers. Furthermore, no knowledge of relaxations or constraints used in controller design is incorporated in structure characterization (e.g., enforcing a block-diagonal Lyapunov matrix \cite{Sturz-et-al-2021-Distributed-Control-Design-Heterogeneous-Interconnected}). Under a supposedly stabilizing structure, it might be impossible to find a controller that actually stabilizes the system in closed-loop using a desired control method. We characterize such approaches as \emph{control-agnostic}. Further control-agnostic approaches including community detection \cite{Daoutidis-et-al-2018-Decomposing-Complex-Plants} are reviewed in the survey \cite{Chanfreut-et-al-2021-Survey-Clustering-Methods}.\par

An alternative approach is to design the communication topology and distributed controller together within a single problem, such as in mixed-integer structured control synthesis \cite{Gross-2011-Optimized-Distributed-Control-Network-Topology-Design}, sparsity-promoting algorithms \cite{Lin-2013-Design-Optimal-Sparse-Feedback-Gains-via-ADMM,Dorfler-et-al-2014-Sparsity-Promoting-Optimal-WAC,Schuler-2014-PhD-Thesis}, and coalitional MPC schemes \cite{Maestre-et-al-2014-A-Coalitional-Control-Scheme-Cooperative-Game}. Whilst control performance is likely to be high quality, these monolithic approaches may lead to computational challenges in large systems and are not applicable if the control agents do not share their cost parameters with the topology designer. As topology and controller design are combined in a single problem, we characterize these methods as \emph{co-design} approaches.
Direct data-driven control methods based on Willems' Fundamental Lemma \cite{Willems-2005-Fundamental-Lemma}, where control decisions are computed directly from data thus avoiding model identification, have gained traction recently \cite{Berberich-Allgower-Review-Data-Driven-MPC,Coulson-2019-Shallows-of-the-DeePC,De-Persis-2020-Formulas-Data-Driven-Control}. Works to-date on topology design for data-driven distributed control focus on matching the communication topology to an identified underlying physical network graph \cite{van-Waarde-2021-Topology-Identification-Networks,Steentjes-2021-H_inf-Performance-Analysis-Distributed-Controller-Synthesis-Interconnected-Linear-Sys-Noisy-Input-State-Data}. We aim to couple the topology and control design problems and allow for greater flexibility than simply mirroring the physical topology in the cyber layer.
\paragraph*{Main Contributions}
In this paper, we present a co-design scheme to synthesize a communication topology and accordingly structured \(\mathcal{H}_2\) controller directly from data. Secondly, we propose a sequential \emph{control-aware} scheme, where the communication designer applies relaxations used in control design within the topology design process, without explicitly optimizing the controller over the closed-loop control cost. A stabilizing control law is guaranteed to exist for the optimized topology assuming that the controller is designed in a fashion that is compatible with the control-aware stability condition. Both approaches account for link costs and noisy data, permit communication topologies that differ from the physical interconnection topology, and certify the existence of a compatible stabilizing controller under the stated assumptions. Our flexible design methods are in the form of mixed-integer semidefinite programs (MISDPs).

\paragraph*{Notation}
The $a \times b$ matrix of ones is denoted by $\mathbbm{1}^{a \times b}$, whilst the symbol \(\land\) denotes the logical and.
For an $N$-sample signal $\{x(k)\}_{k=1:N}$ with $x(k) \in \mathbb{R}^n$, we denote the vectorized signal $\begin{bmatrix}
    x(1)^\top & \ldots & x(N)^\top
\end{bmatrix}^\top$ by $x_{[1:N]} \in \mathbb{R}^{nN}$. To represent the vertical concatenation of vectors we write $\begin{bmatrix}
    x^\top_1 & \ldots & x^\top_N
\end{bmatrix}^\top = [x_i]_{i=1:N}$, where $x_i \in \mathbb{R}^{n_i} \ \forall i = 1,\ldots,N$. A block-diagonal matrix $A$ composed of $N$ diagonal blocks $A_1, \ldots, A_N$ is denoted $A = \operatorname{diag} \, \{A_i\}_{i=1}^N$. A matrix $A$ formed of elements or block elements $A_{ij}$, where $i \in \mathcal{I}_1$ denotes row positioning and $j \in \mathcal{I}_2$ denotes column positioning, is denoted $A = [A_{ij}]_{i \in \mathcal{I}_1}^{j \in \mathcal{I}_2}$. If $\mathcal{I}_1 = \mathcal{I}_2 = \{1,...,N\}$, we may write $[A_{ij}]_{i,j=1:N}$. 

\section{Preliminaries}
\label{sec:preliminaries}


Consider a linear time-invariant (LTI) dynamical system composed of $M \in \mathbb{Z}_+$ interconnected subsystems, each with an individual controlling agent. We denote the sets of subsystems and agents by $\mathcal{V} \coloneq \{1,\ldots,M\}$. The fixed physical interconnection topology is described by the physical-layer graph $\mathcal{G}_P \coloneq (\mathcal{V}, \mathcal{E}_P)$, where \((j,i) \in \mathcal{E}_P\) is an edge of \(\mathcal{G}_P\) if the state of subsystem \(i\) is directly influenced by subsystem \(j\). The communication topology is defined by the design-dependent cyber-layer graph \(\mathcal{G}_C \coloneq (\mathcal{V}, \mathcal{E}_C)\).
There exists an edge $(j,i) \in \mathcal{E}_C$ if agent \(i\) receives an input-output data stream from agent \(j\).
To indicate whether a link exists from agent $j$ to agent $i$, we define a Boolean variable $\delta_{ij} \in \{0,1\} \ \forall i,j \in \mathcal{V}$. The communication topology is represented by a matrix $\Delta \coloneq [\delta_{ij}]_{i,j=1:M}$. The edge set \(\mathcal{E}_C\) is parameterized by \(\Delta\):
\begin{equation}
    \label{eqn:define_delta_ij}
    \mathcal{E}_C = \{(j,i) \in \mathcal{V} \times \mathcal{V} : \delta_{ij} = 1\}
\end{equation}
We associate a non-negative scalar link cost, $c_{ij} \in \mathbb{R}_{\geq 0}$, with any possible communication link $(j,i) \in \mathcal{V} \times \mathcal{V}$ (a possible edge in \(\mathcal{G}_C\)).
The link costs may represent the bandwidth, energy or infrastructure expenditure required to enable the links and maintain the communication network at a desired quality-of-service.
The physical dynamics of subsystem $i$ are described by:
{\allowdisplaybreaks
\begin{subequations}
\label{eqn:subsys_dynamics}
\begin{align}
    \label{eqn:subsys_state_dyn}
    x_i(k+1) &= \sum_{j=1}^M {A_{ij} x_j(k)} + B_i u_i(k) + F_i \varepsilon_i(k)\\
    \label{eqn:subsys_output_eqn}
    y_i(k) &= C_i x_i(k) + D_i u_i(k)\\
    \hat{y}_i(k) &= y_i(k) + v_i(k),
    \end{align}
\end{subequations}}
where, for subsystem $i \in \mathcal{V}$ and at time instant $k$, $x_i(k) \in \mathbb{R}^{n_i}$ is the subsystem state, $u_i(k) \in \mathbb{R}^{m_i}$ is the control input, $y_i(k) \in \mathbb{R}^{p_i}$ is the true output, $\hat{y}_i(k)$ is the noisy output measurement, $\varepsilon_i(k) \in \mathbb{R}^{n_{d,i}}$ is a stochastic process disturbance and $v_i(k) \in \mathbb{R}^{p_i}$ is measurement noise. We have \(A_{ij} = \mathbf{0}^{n_i \times n_j} \; \forall (j,i) \notin \mathcal{E}_P, \ i \neq j\). By stacking the local input, state, output, disturbance and noise vectors and appropriately concatenating the dynamics matrices, the global system dynamics can be written as:
\begin{subequations}
\label{eqn:global_dynamics}
\begin{align}
    \label{eqn:global_state_dyn}
    x(k+1) &= A x(k) + B u(k) + F \varepsilon(k) \\
    \label{eqn:global_output_eqn}
    y(k) &= C x(k) + D u(k), \ \hat{y}(k) = y(k) + v(k).
\end{align}
\end{subequations}
We make the following assumption that the parameters of the state space matrices may be unknown:
\begin{assum}
    \label{assum:LTI_dynamics_unknown}
    The subsystem interconnections are described by the graph $\mathcal{G}_P$, and the pairs $(A,C)$ and $(A,B)$ are observable and controllable, respectively.
    We assume that, $\forall i \in \mathcal{V}$, the matrices $B_i$, $C_i$ and $D_i$, and $\forall i,j \in \mathcal{V}$ $A_{ij}$, in \eqref{eqn:subsys_dynamics} may be unknown. The system lag \cite{Berberich-et-al-2021-Design-of-Terminal-Ingredients}, \(\ell\), is assumed to be known, and the condition \(p \ell = n\) holds, where \(p = \sum_{i=1}^M p_i\) and \(n = \sum_{i=1}^M n_i\). 
\end{assum}


{\renewcommand{\arraystretch}{1}
Each agent $i$ determines $u_i(k)$ based on its available information. This includes local inputs and outputs, and inputs and outputs received via communication. We do not assume state measurement availability and avoid the problem of observer design by working with an extended state $\xi_i(k) \coloneq \begin{bmatrix}
    u_{i,[k-\ell:k-1]}^{\top} & \hat{y}_{i,[k-\ell:k-1]}^{\top}
\end{bmatrix}^\top$. The global extended state vector at time $k$ is given by $\xi(k) = \begin{bmatrix}
    u_{[k-\ell:k-1]}^{\top} & \hat{y}_{[k-\ell:k-1]}^{\top}
\end{bmatrix}^\top \in \mathbb{R}^{\ell (m + p)}$, where \(m = \sum_{i=1}^M m_i\). We assume linear feedback control; the global input is calculated by \(u(k) = K \xi(k)\). If Assumption \ref{assum:LTI_dynamics_unknown} holds, then in the nominal noise-free case, $x(k) \in \mathbb{R}^n$ may be reconstructed from $\xi(k)$ \cite{Koch-et-al-2022-Provably-Robust-Verification-Dissipativity}.} We write the extended state dynamics as:
\begin{equation}
    \label{eqn:extended_state_dynamics}
    \xi(k+1) = \tilde{A} \xi(k) + \tilde{B}u(k) + E_d \tilde{d}(k),
\end{equation}
where $\tilde{d}(k) \in \mathbb{R}^{p}$ is a disturbance that can be constructed from \eqref{eqn:global_dynamics} and \(\varepsilon_{[k-\ell:k-1]}\), \(v_{[k-\ell:k]}\), \(E_d = \begin{bmatrix}
    0 & I
\end{bmatrix}^\top\), and \(\tilde{A}\), \(\tilde{B}\) can be constructed from \eqref{eqn:global_dynamics} -- see \cite{Berberich-et-al-2021-Design-of-Terminal-Ingredients,Koch-et-al-2022-Provably-Robust-Verification-Dissipativity} for details.
We separate the known and unknown dynamics: \(\tilde{A} = A' + \bar{A}\), \(\tilde{B} = B' + \bar{B}\), where \(A'\) and \(B'\) are the known parts of \(\tilde{A}\) and \(\tilde{B}\).
This yields a linear fractional transformation (LFT) representation \cite{Berberich-et-al-2023-Prior-Knowledge-Data-Robust-Design}:
{\allowdisplaybreaks
{\renewcommand{\arraystretch}{1}
\begin{subequations}
\label{eqn:LFT_extended_state}
    \begin{align}
    \label{eqn:LFT_dynamics}
        \left[
        \begin{array}{c}
            \xi(k+1) \\
            \hline e(k) \\ z(k)
        \end{array}
        \right] &= \left[\begin{array}{c|c c c}
            A' & B' & E_d & B_w \\
            \hline C_e & D_{e} & 0 & 0 \\ C_z & D_z & 0 & 0
        \end{array}
        \right] \left[
        \begin{array}{c}
            \xi(k) \\
            \hline u(k) \\ \tilde{d}(k) \\ w(k)
        \end{array}
        \right] \\
        w(k) &= \Psi_{\mathrm{tr}} z(k),
    \end{align}
\end{subequations}}}

where \(w(k) \in \mathbb{R}^{n_w}\) and \(z(k) \in \mathbb{R}^{n_z}\) represent an uncertainty channel, and \(\Psi_{\mathrm{tr}}\) is the true value of the uncertainty (i.e., \(B_w \Psi_\mathrm{tr} (C_z \xi(k) + D_z u(k)) = \bar{A} \xi(k) + \bar{B} u(k)\)). The LFT representation consists of the LTI system \eqref{eqn:extended_state_dynamics} with performance signal \(e(k)\) in feedback with the uncertainty \(\Psi_{\mathrm{tr}}\). The known matrices \(B_w\), \(C_z\) and \(D_z\) describe how the uncertainty enters the dynamics.
To deal with the unknown dynamics, we take a data-driven approach where we utilize a recorded input-output \(T\)-sample dataset. We denote an input signal $\{u^d(k)\}_{k=0:T-1}$ and the corresponding measured output signal $\hat y_{[0:T-1]}^{d} \in \mathbb{R}^{pT}$. Define data matrices as follows:
\begin{subequations}
\begin{align}
    \Xi \coloneq& \begin{bmatrix}
    \xi^{d}(\ell) & \ldots & \xi^{d}(T-1)
\end{bmatrix} \\
\Xi_+ \coloneq& \begin{bmatrix}
    \xi^{d}(\ell+1) & \ldots & \xi^{d}(T)
\end{bmatrix} \\
U \coloneq& \begin{bmatrix}
    u^{d}(\ell) & \ldots & u^{d}(T-1)
\end{bmatrix} \\
Z \coloneq& C_z \Xi + D_z U, \quad W \coloneq \Xi_+ - A' \Xi - B' U\\
\tilde{D}_\mathrm{tr} \coloneq& \begin{bmatrix}
    \tilde{d}^{d}(\ell) & \ldots & \tilde{d}^{d}(T-1)
\end{bmatrix},
\end{align}
\end{subequations}
where \(\tilde{D}_\mathrm{tr}\) is an unknown disturbance matrix.

\section{Problem Formulation}
\label{section:framework}

\subsection{Co-Design of Topology and Controller}
\label{subsec:co-design_problem_formulation}

We can describe the overall communication cost of a topology \(\Delta\) by \(\sum_{i,j=1}^M c_{ij}\delta_{ij}\); \(c_{ij}\) is added to the cost if there is a link from agent \(j\) to agent \(i\). We denote a control cost function by \(J(\cdot)\). By minimizing a combined cost subject to stability, structural and performance constraints, we form a co-design problem, \(\mathcal{P}_{\mathrm{Co}\mhyphen\mathrm{Des}}\), enabling a trade-off between communication and control costs:
\allowdisplaybreaks{
\begin{subequations}
\label{eqn:comms_control_co-design}
    \begin{align}
    \mathcal{P}_{\mathrm{Co}\mhyphen\mathrm{Des}}: \quad &\min_{P \succ 0, K, \Delta} J(P,K) + \sum_{i,j=1}^M c_{ij}\delta_{ij} \\
        \label{eqn:co-des_stability_constraint}
        \mathrm{s.t.} \ & \mathcal{C}_\mathrm{stab}(P,K,\mathcal{D},\mathcal{M}) \\
        \label{eqn:co-des_structure_constraint}
        & \mathcal{C}_\mathrm{str}(P,K,\Delta) \\
        \label{eqn:co-des_perf_constraint}
        & \mathcal{C}_\mathrm{perf}(P,K,\mathcal{D},\mathcal{M}) \\
    \delta_{ij} \in \{0,1\}  &\ \forall i,j \in \mathcal{V}, \quad \Delta = [\delta_{ij}]_{i,j=1:M}
    \end{align}
\end{subequations}}

We optimize over the controller \(K\), Lyapunov matrix \(P\), and communication topology \(\Delta\). The problem is parameterized by the dataset \(\mathcal{D}\) and available model knowledge \(\mathcal{M} \coloneq \{A',B'\}\). One may seek to optimize the control performance subject to a communication budget \(\bar{c}\), which can be formulated by removing \(\sum_{i,j=1}^M c_{ij}\delta_{ij}\) from the cost function and adding a constraint \(\sum_{i,j=1}^M c_{ij}\delta_{ij} \leq \bar{c}\). We require a data-driven stability constraint \(\mathcal{C}_\mathrm{stab}(P,K,\mathcal{D},\mathcal{M})\) \eqref{eqn:co-des_stability_constraint} to ensure closed-loop stability, and we allow for performance constraints \(\mathcal{C}_\mathrm{perf}(P,K,\mathcal{D},\mathcal{M})\) \eqref{eqn:co-des_perf_constraint} to guarantee an upper bound on, for example, control performance as measured by the closed-loop \(\mathcal{H}_2\)-norm. The structure constraint \(\mathcal{C}_\mathrm{str}(P,K,\Delta)\) \eqref{eqn:co-des_structure_constraint} ensures that the controller has a structure consistent with the communication topology \(\Delta\). We define selector matrices \(E_{u,i}, \,E_{\xi,i} \ \forall i \in \mathcal{V}\) such that \(u_i(k) = E_{u,i} u(k)\) and \(\xi_i(k) = E_{\xi,i} \xi(k)\). The sub-controller corresponding to agent \(j\)'s contribution to \(u_i(k)\) is given by \(K_{ij} = E_{u,i} K E_{\xi,j}^\top\). We require \(K_{ij} = 0\) if \((j,i) \notin \mathcal{E}_C\), which means that agent \(i\) does not use information from agent \(j\) in calculating its control decision if \((j,i) \notin \mathcal{E}_C\).
Using \eqref{eqn:define_delta_ij}, this can be expressed as:
\begin{align}
\label{eqn:struc_implies_constraint_def}
    \delta_{ij} = 0 \implies K_{ij} = E_{u,i} K E_{\xi,j}^\top = 0 \ \forall i,j \in \mathcal{V}.
\end{align}
Any useful controller will stabilize the system in closed-loop. 
Input-to-state stability of the extended state dynamics \eqref{eqn:extended_state_dynamics} with respect to $\tilde d$, under $u(k)=K\xi(k)$, is guaranteed by the existence of a $P\succ0$ satisfying \cite{Scherer-Weiland-2011-LMIs-in-Control}
\begin{equation}
    \label{eqn:Lyap_stability_condition}
    (\tilde A + \tilde BK) P^{-1} (\tilde A + \tilde BK)^\top - P^{-1} + E_d E_d^\top \prec 0.
\end{equation}
The stability constraint \(\mathcal{C}_\mathrm{stab}(P,K,\mathcal{D},\mathcal{M})\) must be such that satisfaction of \(\mathcal{C}_\mathrm{stab}(P,K,\mathcal{D}, \mathcal{M}) \implies \eqref{eqn:Lyap_stability_condition}\) holds, and is desirably convex. Similarly, we require a (convex) \(\mathcal{C}_\mathrm{str}(P,K,\Delta)\) such that satisfaction of \(\mathcal{C}_\mathrm{str}(P,K,\Delta) \implies \eqref{eqn:struc_implies_constraint_def}\) holds. The joint constraints \eqref{eqn:struc_implies_constraint_def} and \eqref{eqn:Lyap_stability_condition} are non-convex in \(K\) and \(P\); different relaxations will lead to different feasible sets in the relaxed problems. This warrants careful consideration with regard to the structured stabilization problem feasibility and hence the presence of unstable SFMs. We propose the generalization of SFMs to \emph{design-dependent SFMs}.
\begin{rem}
\label{rem:design-depend-SFMs}
    Unstable design-dependent SFMs for a given \(\Delta\) denote infeasibility of the chosen sufficient synthesis conditions under the feedback structure associated with \(\Delta\). However, absence of unstable classic SFMs for the plant and feedback structure is necessary, but not sufficient, for structured synthesis feasibility.
\end{rem}

\subsection{Sequential Topology and Controller Design}

In a sequential approach, the communication topology is designed first, before the controller is synthesized subject to the topology structure. 
We propose to couple the control synthesis and topology design problems within a sequential control-aware scheme. We use compatible stability and structural constraints in both stages, so that every feasible topology design solution supplies a feasible controller for the subsequent controller synthesis stage.
The control-aware communication design problem \(\mathcal{P}_{\mathrm{Ctrl}\mhyphen\mathrm{Aw}}\), including an optional \(\mathcal{C}_\mathrm{perf}\), is given by:
\begin{subequations}
\label{eqn:ctrl_aw_framework_optim_formulation}
    \begin{align}
        &\mathcal{P}_{\mathrm{Ctrl}\mhyphen\mathrm{Aw}}: \ \ \min_{P \succ 0,K,\Delta} g(P,K,\mathcal{D},\Delta) + \sum_{i,j=1}^M c_{ij}\delta_{ij} \\
        & \qquad \mathrm{s.t.} \ \mathcal{C}_\mathrm{perf}(P,K,\mathcal{D},\mathcal{M}) \\
        \label{eqn:P_SDS_def}
         &\mathcal{P}_{\mathrm{SDS}}:\begin{cases}
            \mathcal{C}_\mathrm{stab}(P,K,\mathcal{D},\mathcal{M}) \\ 
            \mathcal{C}_\mathrm{str}(P,K,\Delta) \\
            \delta_{ij} \in \{0,1\}  \ \forall i,j \in \mathcal{V}, \ \Delta = [\delta_{ij}]_{i,j=1:M}.
        \end{cases}
    \end{align}
\end{subequations}
Considering the feasibility of stabilization, we can describe a structure design and stabilization problem as \(\mathcal{P}_{\mathrm{SDS}}\), where we aim to find a \(P \succ 0, K, \Delta\) that satisfy \eqref{eqn:P_SDS_def}.
Any \(\Delta_\mathrm{feas}\) solving \(\mathcal{P}_\mathrm{SDS}\) represents a communication topology with zero unstable design-dependent SFMs, and any \(K_\mathrm{feas}\) solving \(\mathcal{P}_\mathrm{SDS}\) stabilizes the system dynamics \eqref{eqn:global_dynamics}. We let \(g(\cdot)\) be a control cost (proxy) function in the relaxed variables representing the communication designer's control goal. By construction of the topology design constraints, the subsequent control synthesis problem is guaranteed to be feasible.

\begin{rem}
    If the objective \(g(P, K, \mathcal{D},\Delta) = J(P,K)\), the communication design stage reduces to \(\mathcal{P}_{\mathrm{Co}\mhyphen\mathrm{Des}}\) \eqref{eqn:comms_control_co-design}. Alternatively, we can choose a different cost function \(g(\cdot)\) to reflect that the communication designer has a different objective than the control designer, does not know the true control cost so uses a proxy function, or simplifies the objective to reduce the computational burden.
\end{rem}

\section{Stability and Structure Constraints}
\label{section:constraint_formulation}

We formulate a stability constraint based on a data-driven control approach \cite{Berberich-et-al-2023-Prior-Knowledge-Data-Robust-Design} and utilize structural constraints in \cite{Gross-2011-Optimized-Distributed-Control-Network-Topology-Design} and \cite{Ferrante-et-al-2020-Design-of-Structured-Stabilizers}. We assume that the unknown disturbance matrix \(\tilde{D}_\mathrm{tr}\) is bounded by some known set:

\begin{assum}
    \label{assum:Berberich_noise_bounds}
    {\renewcommand{\arraystretch}{1}\cite{Berberich-et-al-2023-Prior-Knowledge-Data-Robust-Design} A convex cone of symmetric matrices admitting an LMI representation \(\mathbf{P}_d\) is known, and is such that the disturbance matrix satisfies \(\tilde{D}_\mathrm{tr} \in \mathbf{\tilde{D}}\), where 
    \begin{equation}
        \mathbf{\tilde{D}} \coloneq \left\{ \tilde{D} \, \left| \, \begin{bmatrix}
        \tilde{D}^\top \\ I
    \end{bmatrix}^\top P_d \begin{bmatrix}
        \tilde{D}^\top \\ I
    \end{bmatrix} \succeq 0 \ \forall P_d \in \mathbf{P}_d  \right. \right\}.
    \end{equation}}
\end{assum}

We assume that a multiplier combining prior knowledge, the disturbance bound, and information embedded within the dataset is known, which bounds an uncertainty set \(\mathbf{\tilde{\Psi}}_\mathrm{com}\) containing \(B_w \Psi_\mathrm{tr}\).

\begin{assum}
\label{assum:combined_uncertainty_set}
    {\renewcommand{\arraystretch}{1}\cite{Berberich-et-al-2023-Prior-Knowledge-Data-Robust-Design} An uncertainty set is known and is such that \(B_w \Psi_\mathrm{tr} \in \mathbf{\tilde{\Psi}}_\mathrm{com}\). The set admits a multiplier description: 
    \begin{equation}
        \mathbf{\tilde{\Psi}}_\mathrm{com} \coloneq \Bigg\{ \tilde{\Psi} \, \left| \, \begin{bmatrix}
        \tilde{\Psi}^\top \\ I
    \end{bmatrix}^\top \tilde{P}_\mathrm{com} \begin{bmatrix}
        \tilde{\Psi}^\top \\ I
    \end{bmatrix} \succeq 0 \right. \forall \tilde{P}_\mathrm{com} \in \mathbf{\tilde{P}}_\mathrm{com} \Bigg\},
    \end{equation}
    where \(\mathbf{\tilde{P}}_\mathrm{com}\) is a known convex cone of symmetric matrices admitting an LMI representation. Furthermore, \(\mathbf{\tilde{\Psi}}_\mathrm{com} = \mathbf{\tilde{\Psi}}_\mathrm{pr} \cap \mathbf{\tilde{\Psi}}_\mathrm{learnt}\), where \(\mathbf{\tilde{\Psi}}_\mathrm{pr}\) is an a-priori known set given by prior system knowledge (e.g., a norm bound), and \(\mathbf{\tilde{\Psi}}_\mathrm{learnt}\) is inferred from data. In particular, \(\mathbf{\tilde{\Psi}}_\mathrm{learnt} \coloneq \left\{ \tilde{\Psi} \, \left| \, \begin{bmatrix}
        \tilde{\Psi}^\top \\ I
    \end{bmatrix}^\top \tilde{P}_d \begin{bmatrix}
        \tilde{\Psi}^\top \\ I
    \end{bmatrix} \succeq 0 \ \forall \tilde{P}_d \in \mathbf{\tilde{P}}_d  \right. \right\}\), where \(\mathbf{\tilde{P}}_d \coloneq \begin{bmatrix}
        -Z^\top & W^\top \\
        0 & E_d^\top
    \end{bmatrix}^\top \mathbf{P}_d \begin{bmatrix}
        -Z^\top & W^\top \\
        0 & E_d^\top
    \end{bmatrix}\).
    }
\end{assum}

We refer to \cite{Berberich-et-al-2023-Prior-Knowledge-Data-Robust-Design} for details on how to construct \(\mathbf{P}_d\) and \(\mathbf{\tilde{P}}_\mathrm{com}\) from possible disturbance descriptions and descriptions of prior knowledge on the uncertainty.

\subsection{Data-Driven Stability Constraint (\(\mathcal{C}_\mathrm{stab}\), \eqref{eqn:co-des_stability_constraint})}
\label{subsec:data-driven-stab-constraint}



Standard control synthesis LMIs use a parameterization of the form \(KX = L\) to convexify the problem, where \(X = P^{-1} \succ 0\) is the inverse of the Lyapunov matrix. To decouple the controller parameterization from the Lyapunov matrix and avoid indirectly imposing sparsity constraints on \(P\) within \(\mathcal{C}_\mathrm{str}\), we apply an extended stability parameterization proposed in \cite{De-Oliveira-2002-Extended-H2-H-inf-Norm-Characterizations-Controller-Parameterizations} and instead use \(KG = L\), where \(G + G^\top - X \succ 0\), thus enlarging the feasible set of the design problem.

\begin{thm}
\label{theorem:extended_data-driven_stability}
    Consider the extended state dynamics \eqref{eqn:extended_state_dynamics} and suppose Assumptions \ref{assum:LTI_dynamics_unknown}, \ref{assum:Berberich_noise_bounds} and \ref{assum:combined_uncertainty_set} hold. Then the following LMI, in the variables \(X \succ 0\), \(\tilde{P}_\mathrm{com} \in \mathbf{\tilde{P}}_\mathrm{com}\), \(L\) and \(G\), represents a valid stability constraint \(\mathcal{C}_\mathrm{stab}\) (cf. \eqref{eqn:co-des_stability_constraint}) for topology co-design:
\begin{equation}
\label{eqn:extended_stability_data_driven_LMI}
\begin{bmatrix}
		        {\renewcommand{\arraystretch}{1}\begin{bmatrix}
            E_d E_d^\top - X & 0 \\ 0 & 0
        \end{bmatrix} + \begin{bmatrix}
            0 & I \\ I & 0
        \end{bmatrix}^\top \tilde{P}_\mathrm{com} \begin{bmatrix}
            0 & I \\ I & 0
        \end{bmatrix}} & \star \\
        {\renewcommand{\arraystretch}{1}\begin{bmatrix}
            (A' G + B'L)^\top & (C_z G + D_z L)^\top
        \end{bmatrix}} & -\chi
\end{bmatrix} \prec 0,
\end{equation}
where \( \chi = (G + G^\top - X) \). If a feasible solution to \eqref{eqn:extended_stability_data_driven_LMI} exists, then \(P = X^{-1}\) and a controller \(K = LG^{-1}\) satisfy \eqref{eqn:Lyap_stability_condition}, and the dynamics \eqref{eqn:extended_state_dynamics} are input-to-state stable with respect to \(\tilde{d}\) under the feedback control law \(u(k) = K \xi(k)\).
\end{thm}
\begin{pf}
We prove that \eqref{eqn:extended_stability_data_driven_LMI} is feasible if and only if the stability LMI associated with \cite[Thm. 1]{Berberich-et-al-2023-Prior-Knowledge-Data-Robust-Design} is feasible, which directly proves our result. To prove sufficiency, consider that \(G^\top X^{-1} G \succeq G + G^\top - X \) since \(X^{-1} \succ 0\) \cite{De-Oliveira-2002-Extended-H2-H-inf-Norm-Characterizations-Controller-Parameterizations}. Therefore, replacing \(-\chi\) in \eqref{eqn:extended_stability_data_driven_LMI} by \(-G^\top X^{-1}G\) preserves negative definiteness. Since \(\chi\succ0\) for feasibility, \(G+G^\top=\chi+X\succ0\), which implies that \(G\) is nonsingular. Pre- and post-multiplying by \(\operatorname{diag} \,(I,(G^{-1}X)^\top)\) and \(\operatorname{diag} \,(I,G^{-1}X)\), respectively, then yields:
\begin{equation}
    \label{eqn:berberich_stability_LMI}
    \begin{bmatrix}
        {\renewcommand{\arraystretch}{1}\begin{bmatrix}
            E_d E_d^\top - X & 0 \\ 0 & 0
        \end{bmatrix} + \begin{bmatrix}
            0 & I \\ I & 0
        \end{bmatrix}^\top \tilde{P}_\mathrm{com} \begin{bmatrix}
            0 & I \\ I & 0
        \end{bmatrix}} & \star \\
        {\renewcommand{\arraystretch}{1}\begin{bmatrix}
            X A'^\top + \bar L^\top B'^\top & X C_z^\top + \bar L^\top D_z^\top
        \end{bmatrix}} & -X
    \end{bmatrix} \prec 0,
\end{equation}
where \(\bar L=LG^{-1}X\), which is equivalent to (23) in \cite{Berberich-et-al-2023-Prior-Knowledge-Data-Robust-Design}. The recovered controller is \(\bar L X^{-1}=LG^{-1} = K\). 
Conversely, given a feasible solution \((X,\bar L,\tilde P_{\mathrm{com}})\) of \eqref{eqn:berberich_stability_LMI}, choose \(G=X\) and \(L=\bar L\). Then \(\chi=X\), and \eqref{eqn:extended_stability_data_driven_LMI} coincides with \eqref{eqn:berberich_stability_LMI}, thus proving necessity.
Suppose Assumptions \ref{assum:LTI_dynamics_unknown}, \ref{assum:Berberich_noise_bounds} and \ref{assum:combined_uncertainty_set} hold. Then by \cite[Thm. 1]{Berberich-et-al-2023-Prior-Knowledge-Data-Robust-Design}, a feasible solution to \eqref{eqn:berberich_stability_LMI} 
implies that \(
(\tilde A+\tilde BK)X(\tilde A+\tilde BK)^\top
-X+E_dE_d^\top\prec0
\) for every system in the specified uncertainty set, where
$K=\bar L X^{-1}=LG^{-1}$.
Setting $P=X^{-1}$ yields
\eqref{eqn:Lyap_stability_condition}.
Thus, feasibility of
\eqref{eqn:extended_stability_data_driven_LMI}
guarantees input-to-state stability of the extended-state
closed-loop dynamics with respect to $\tilde d$ and establishes
a valid stability constraint $\mathcal C_{\mathrm{stab}}$.
\qed
\end{pf}


\begin{rem}
    The condition \(p \ell = n\) is necessary to apply the extended state-based control design approach we employ \cite{Berberich-et-al-2021-Design-of-Terminal-Ingredients}; in general this condition is restrictive since $p \ell \geq n$. Recent works \cite{Alsalti-et-al-2023-Non-minimal-state,Li-2026-Controller-Synthesis-Noisy-I-O-Data} aim to relax this restriction. Note that \cite{Alsalti-et-al-2023-Non-minimal-state} does not consider system structure, whilst \cite{Li-2026-Controller-Synthesis-Noisy-I-O-Data} may be more capable of retaining system structure, but does not provide a general procedure for constructing an augmentation satisfying the required conditions.
\end{rem}

\subsection{Structure Constraint Formulation (\(\mathcal{C}_\mathrm{str}\), \eqref{eqn:co-des_structure_constraint})}
\label{subsec:structure-constraint}


To impose the desired structure on \(K\), we require conditions on \(L\) and \(G\) such that \eqref{eqn:struc_implies_constraint_def} holds for \(K = LG^{-1}\). There is no known convex formulation to exactly express this unless possibly restrictive assumptions hold such as quadratic invariance \cite{Rotkowitz-Lall-2006-Characterization-Convex-Problems-Decentralized-Control} or chordal sparsity \cite{Watanabe-et-al-2024-Convex-Reformulation-LMI-Based-Distributed-Controller-Design-Class-of-Non-Block-Diagonal-Lyapunov-Functions}. 
Approaches given in \cite{Gross-2011-Optimized-Distributed-Control-Network-Topology-Design,Ferrante-et-al-2020-Design-of-Structured-Stabilizers} provide sufficient conditions on \(L\) and \(G\), allowing for general sparsity patterns at the cost of conservatism.  
In fact, the conditions given in these two contributions are equivalent. The proof is not included here due to space restrictions. We choose to use the approach in \cite{Gross-2011-Optimized-Distributed-Control-Network-Topology-Design} rather than \cite{Ferrante-et-al-2020-Design-of-Structured-Stabilizers} as no additional optimization variables are introduced. Since each control agent can access local information without requiring communication, we set \(\delta_{ii} = 1 \ \forall i \in \mathcal{V}\). The sufficient conditions in \cite[Thm. 2]{Gross-2011-Optimized-Distributed-Control-Network-Topology-Design} then reduce to:
{\allowdisplaybreaks
\begin{subequations}
\label{eqn:Gross-structure-constraints}
    \begin{align}
    \label{eqn:gross_M_mat_structure_constraint}
         \delta_{ij} = 0 &\implies E_{u,i} L E_{\xi,j}^\top = 0 \ \forall i,j
    \\
    \label{eqn:gross_X_mat_structure_constraint}
        \delta_{ij} = 0 & \ \land \ \delta_{iz} = 1 \implies E_{\xi,z} G E_{\xi,j}^\top = 0 \ \forall i,j,z.
    \end{align}
\end{subequations}}

Since satisfaction of \eqref{eqn:Gross-structure-constraints} \(\implies\) \eqref{eqn:struc_implies_constraint_def} holds, we take \eqref{eqn:Gross-structure-constraints} as our structural constraint \(\mathcal{C}_\mathrm{str}\).

\section{Data-Driven Topology Design Schemes}
\label{section:co-design}


\subsection{Co-Design of Topology and \(\mathcal{H}_2\) Controller}

The control and communication co-design approach may be desirable if the topology designer knows the system-wide control goal and has sufficient computational resources to explicitly optimize over the closed-loop control cost. Our formulation represents an extension of \cite{Miller-2025-Data-Driven-Structured-Robust-Control} to topology design and solves the co-design problem \(\mathcal{P}_{\mathrm{Co}\mhyphen\mathrm{Des}}\) \eqref{eqn:comms_control_co-design}. 
We find a performance bound on the \(\mathcal{H}_2\)-norm of the transfer function from \(\tilde{d}(k)\) to the performance signal \(e(k) = C_e \xi(k) + D_e u(k)\):
\begin{lem}
\label{lemma:berberich_H2_perf}
    Consider the extended state dynamics \eqref{eqn:extended_state_dynamics} and performance signal \(e(k) = C_e \xi(k) + D_e u(k)\). Suppose that Assumptions \ref{assum:LTI_dynamics_unknown}, \ref{assum:Berberich_noise_bounds} and \ref{assum:combined_uncertainty_set} hold, and a feasible solution to \eqref{eqn:extended_stability_data_driven_LMI} exists. Then the closed-loop \(\mathcal{H}_2\)-norm of \(\tilde{d} \mapsto e\) when \(u(k) = K\xi(k)\) is upper-bounded by \(\gamma_\mathrm{max}\) if there exist a \(\Gamma, X \succ 0\), \(G\) and \(L\) such that \(\operatorname{trace} \, (\Gamma) < \gamma_\mathrm{max}^2\) and the following LMI holds:
    \begin{equation}
    \label{eqn:extended_H2_performance_LMI}
        {\renewcommand{\arraystretch}{1}\begin{bmatrix}
            \Gamma & C_e G + D_e L \\ \star & G + G^\top - X
        \end{bmatrix}} \succ 0,
    \end{equation}
    with \(K = LG^{-1}\), and \(X\), \(G\) and \(L\) solving \eqref{eqn:extended_stability_data_driven_LMI}.
\end{lem}
\begin{pf}
    The result follows as an application of \cite[Thm. 5]{De-Oliveira-2002-Extended-H2-H-inf-Norm-Characterizations-Controller-Parameterizations} to a result in \cite[Thm. 1]{Berberich-et-al-2023-Prior-Knowledge-Data-Robust-Design}. \qed
\end{pf}
For a given \(\gamma_\mathrm{max}\), a valid \(\mathcal{H}_2\) performance constraint \(\mathcal{C}_\mathrm{perf}\) is given by \(\operatorname{trace} \, (\Gamma) < \gamma_\mathrm{max}^2\) and \eqref{eqn:extended_H2_performance_LMI}.

\subsubsection{Implementation of the Co-Design Problem \eqref{eqn:comms_control_co-design}}
\label{subsec:co-design_implementation}

To optimize control performance, we minimize an upper bound on the squared closed-loop \(\mathcal{H}_2\)-norm, \(\operatorname{trace} \, (\Gamma)\). Combined with the communication link costs \(c_{ij}\) and structural \eqref{eqn:Gross-structure-constraints}, stability \eqref{eqn:extended_stability_data_driven_LMI}, and performance \eqref{eqn:extended_H2_performance_LMI} constraints, we formulate a complete co-design optimization problem to solve \eqref{eqn:comms_control_co-design}.
\begin{thm}
    \label{theorem:co-design_optimal_H2}
    Consider the extended state dynamics \eqref{eqn:extended_state_dynamics}. Suppose that Assumptions \ref{assum:LTI_dynamics_unknown}, \ref{assum:Berberich_noise_bounds} and \ref{assum:combined_uncertainty_set} hold. Then
    an optimal solution to the following MISDP provides a communication topology \(\Delta^*\) and compatible stabilizing controller \(K^*\), solving 
    \(\mathcal{P}_{\mathrm{Co}\mhyphen\mathrm{Des}}\) \eqref{eqn:comms_control_co-design}, that minimizes the sum of a certified upper bound, \(\operatorname{trace} \, (\Gamma)\), on the squared closed-loop \(\mathcal{H}_2\)-norm of \(\tilde{d} \mapsto e\) and the communication costs:
    {\allowdisplaybreaks
    \begin{subequations}
        \begin{align}
        \label{eqn:bounded_H2_cost_func}
            \min_{X, \Gamma \succ 0, \tilde{P}_\mathrm{com}, L, G, \Delta} & 
            \operatorname{trace} \, (\Gamma) +  \sum_{i,j=1}^M c_{ij} \delta_{ij} \\
            \mathrm{s.t.} \ \Delta = [\delta_{ij}]_{i,j=1:M}, \quad & \delta_{ij} \in \{0,1\} \ \forall i,j \in \mathcal{V}\\
        \label{eqn:extended_H2_performance_LMI_thm}
        {\renewcommand{\arraystretch}{1}\begin{bmatrix}
            \Gamma & C_e G + D_e L \\ \star & G+G^\top - X
        \end{bmatrix}} & \succ 0 \\
        \eqref{eqn:extended_stability_data_driven_LMI}, \ \tilde{P}_\mathrm{com} &\in \mathbf{\tilde{P}}_\mathrm{com} \\
        \label{eqn:big_M_structured_set_constraints_1}
        -\bar{M} \delta_{ij} \mathbbm{1} \leq E_{u,i} & L E_{\xi,j}^\top \leq \bar{M} \delta_{ij} \mathbbm{1} \ \forall (i,j)\\
        \label{eqn:big_M_structured_set_constraints_3}
        \begin{split}
            -\bar{M} (\delta_{ij} - \delta_{iz} + 1) \mathbbm{1} \leq &  \, E_{\xi,z} G E_{\xi,j}^\top \\  \leq \bar{M} (\delta_{ij} - & \delta_{iz} + 1) \mathbbm{1} \ \forall (i,j,z),
        \end{split}
        \end{align}
    \end{subequations}}
    
    where $\mathbbm{1}$ represents a matrix of ones of appropriate dimension, and \(K^* = L^*G^{*-1}\) is structured according to the topology \(\Delta^*\) and stabilizes the system dynamics \eqref{eqn:extended_state_dynamics} in closed-loop.
\end{thm}
\begin{pf}
    We use the Big-M reformulation of the logical implication constraints \eqref{eqn:Gross-structure-constraints} as done in \cite{Gross-2011-Optimized-Distributed-Control-Network-Topology-Design} by choosing a $\bar{M} \in \mathbb{R}_{>0}$ sufficiently large to avoid excluding relevant feasible solutions, resulting in element-wise inequalities \eqref{eqn:big_M_structured_set_constraints_1}-\eqref{eqn:big_M_structured_set_constraints_3}. Any feasible controller \(K = LG^{-1}\) has the desired structure due to the constraints \eqref{eqn:big_M_structured_set_constraints_1} - \eqref{eqn:big_M_structured_set_constraints_3}. If Assumptions \ref{assum:LTI_dynamics_unknown}, \ref{assum:Berberich_noise_bounds} and \ref{assum:combined_uncertainty_set} hold, then by Thm. \ref{theorem:extended_data-driven_stability}, any feasible \(K\) stabilizes the extended state dynamics \eqref{eqn:extended_state_dynamics}. From Lemma \ref{lemma:berberich_H2_perf}, by minimizing \(\operatorname{trace} \, (\Gamma)\), we minimize the squared upper bound of the closed-loop \(\mathcal{H}_2\) norm. Considering the cost function leads to the conclusion that the optimal solution will minimize the sum of this bound and the overall communication costs. \qed
\end{pf}
Thm. \ref{theorem:co-design_optimal_H2} co-optimizes the topology and controller to balance closed-loop \(\mathcal{H}_2\) control performance with communication costs. We note that appropriate choices of \(C_e\) and \(D_{e}\) result in an LQR formulation.

\subsection{Control-Aware Communication Topology Design}
\label{section:control-aware_design}

When different entities design the topology and controller, the centralized co-design formulation is unavailable, motivating a modular design approach.
Moreover, agents may wish to avoid sharing their cost functions with the topology designer. We formulate a control-aware design approach to solve \(\mathcal{P}_{\mathrm{Ctrl}\mhyphen\mathrm{Aw}}\) \eqref{eqn:ctrl_aw_framework_optim_formulation} using the structural and stability constraints developed in Sect. \ref{section:constraint_formulation}. 

\subsubsection{Data-Driven Dynamic Coupling Estimation}
\label{subsec:data-driven_coupling_estimation}

The cost function of \(\mathcal{P}_{\mathrm{Ctrl}\mhyphen\mathrm{Aw}}\) is given by \(g(P,K,\mathcal{D},\Delta) + \sum_{i,j=1}^M c_{ij}\delta_{ij}\). We take \(g(P,K,\mathcal{D},\Delta)\) as a proxy for the closed-loop performance, and propose to approximate the value of communication links in terms of their impact on control cost by analyzing the open-loop coupling strength between pairs of subsystems. The coupling strength is estimated with a data-driven metric based on subspace predictive control (SPC) and robust regression. 
SPC \cite{Favoreel-et-al-1999-SPC} uses a linear output predictor:
{\renewcommand{\arraystretch}{1.0}
\begin{equation}
\label{eqn:spc_predictor}
    y^\mathrm{p}(k) = \Theta \begin{bmatrix}
        \xi(k) \\ u(k)
    \end{bmatrix},
\end{equation}}

where \(y^\mathrm{p}(k) \in \mathbb{R}^p\) is a 1-step output prediction and
$\Theta \in \mathbb{R}^{p \times (m + p)\ell + m}$ is the predictor matrix. Considering an output prediction $y^\mathrm{p}_i(k)$ for subsystem $i$, we use the SPC predictor \eqref{eqn:spc_predictor} to see that {\renewcommand{\arraystretch}{1.0}$y^\mathrm{p}_i(k) = \sum_{j=1}^M \Theta_{ij} \begin{bmatrix}
    \xi_{j}(k) \\ u_{j}(k)
\end{bmatrix}$}. Here \(\Theta_{ij}\) selects the rows for \(y_i\) and columns for (\(\xi_j,u_j\)), in that order.

\begin{prop}
\label{prop:SPC_scaled_error}
    Consider a 1-step output prediction for a subsystem \(i\) made using SPC \eqref{eqn:spc_predictor}. Then the squared Euclidean norm of the prediction change in \(y^\mathrm{p}_i(k)\) resulting from neglecting \(\Theta_{ij}\) for a given \(j\), scaled by \((\lVert \xi_{j}(k) \rVert_2^2 + \lVert u_{j}(k) \rVert_2^2)\), is upper-bounded by \(\lVert \Theta_{ij} \rVert_2^2\).
\end{prop}
\begin{pf}
    The prediction change is given by {\renewcommand{\arraystretch}{1.0}\(\Theta_{ij} \begin{bmatrix}
        \xi_{j}(k) \\ u_{j}(k)
    \end{bmatrix}\)}. Applying the two-norm and using the consistency of an induced matrix norm, we see that the squared Euclidean norm of the prediction error is upper bounded by $\lVert \Theta_{ij} \rVert_2^2 \cdot (\lVert \xi_{j}(k) \rVert_2^2 + \lVert u_{j}(k) \rVert_2^2)$.
    \qed
\end{pf}

Proposition \ref{prop:SPC_scaled_error} motivates a normalized predictor sensitivity metric, used as a heuristic proxy for communication value. It does not bound closed-loop performance.
We define our data-driven metric for coupling strength:
\begin{defn}
    Given a linear predictor $\Theta$ with non-zero diagonal blocks, an estimate of the coupling strength from subsystem $j$ to subsystem $i$ is \(\eta \in \mathbb{R}^{M \times M}\), defined by:
    \begin{equation}
        \label{eqn:coupling_est}
        \eta_{ij} \coloneq \frac{ \lVert \Theta_{ij} \rVert_2^2}{\lVert \Theta_{ii} \rVert_2^2}.
    \end{equation}
    where \(\lVert \cdot \rVert_2\) denotes the induced matrix 2-norm.
\end{defn}
We have normalized by the self-coupling estimates to improve scaling. 
To estimate a structured predictor from data, we form subsystem-ordered data matrices by permuting with the selector matrices: {\renewcommand{\arraystretch}{1}\(\Xi^P = \begin{bmatrix}
    E_{\xi,1}^\top & \ldots & E_{\xi,M}^\top
\end{bmatrix}^\top \Xi\), \(U^F = \begin{bmatrix}
    E_{u,1}^\top & \ldots & E_{u,M}^\top
\end{bmatrix}^\top U\), and write an output data matrix \(Y^F = \begin{bmatrix}
        Y^{F^\top}_1 & \ldots & Y^{F^\top}_M
    \end{bmatrix}^\top\), where $Y_i^F = \begin{bmatrix}
    y^d_i(\ell) & \ldots & y^d_i(T-1)
\end{bmatrix} \ \forall i \in \mathcal{V}$.} The least-squares predictor is the solution to the regression problem \cite{Favoreel-et-al-1999-SPC}:
{\renewcommand{\arraystretch}{1.0}
\begin{equation}
    \label{eqn:spc_LS}
    \min_{\hat{\Theta}} \left\Vert Y^F - \hat{\Theta} \begin{bmatrix}
        \Xi^P \\ U^F
    \end{bmatrix} \right\Vert _F^2,
\end{equation}}

where the predictor columns are permuted consistently with the regressor ordering \(\begin{bmatrix}
    \Xi^{P^\top} & U^{F^\top}
\end{bmatrix}^\top\). To robustify \eqref{eqn:spc_LS} under a stochastic disturbance and output measurement noise, we apply stochastic robust approximation (\cite{Boyd_Vandenberghe_2004-Convex-Optimization-Textbook} Sect. 6.4.1). We calculate \(\Theta = \bar{Y}^F \bar{\Pi}^\top (\bar{\Pi} \bar{\Pi}^\top + \hat{W})^{-1}\), where \(\bar{Y}^F\) and $\bar{\Pi}$ are the empirical averages of \(Y^F\) and $\Pi$ across multiple experiments, respectively, with {\renewcommand{\arraystretch}{1.0}$\Pi = \begin{bmatrix}
    \Xi^{P ^ \top} & U^{F ^ \top}
\end{bmatrix}^\top$} for each experiment, and $\hat{W} = E[(\Pi - \bar{\Pi})(\Pi - \bar{\Pi})^\top]$, which we calculate empirically using the experimental data. We use the same input sequence $u^{d}_{[0:T-1]}$ for each experiment.

\subsubsection{Control-Aware Optimization Formulation for \eqref{eqn:ctrl_aw_framework_optim_formulation}}
\label{subsec:control_aware_optim_formulation}

The concept of control-aware design is based on choosing a stability constraint to ensure that
the second stage structured controller synthesis problem is feasible. 
The structure conditions \eqref{eqn:Gross-structure-constraints} may be combined with data-driven feedback control design to solve \(\mathcal{P}_\mathrm{SDS}\) \eqref{eqn:P_SDS_def} and find a structured stabilizing feedback controller.
\begin{lem}
\label{lemma:gross_stabilizing_structure}
    Consider the dynamics \eqref{eqn:global_dynamics} and the structure design and stabilization problem \(\mathcal{P}_\mathrm{SDS}\) \eqref{eqn:P_SDS_def}. Suppose that Assumptions \ref{assum:LTI_dynamics_unknown}, \ref{assum:Berberich_noise_bounds} and \ref{assum:combined_uncertainty_set} hold. Assume that there exist $X \succ 0$, $\tilde{P}_\mathrm{com} \in \mathbf{\tilde{P}}_\mathrm{com}$, $\Delta \in \{0,1\}^{M \times M}$, $L$, $G$ that satisfy \(\mathcal{C}_\mathrm{str}\) \eqref{eqn:Gross-structure-constraints} and \(\mathcal{C}_\mathrm{stab}\) \eqref{eqn:extended_stability_data_driven_LMI}.
    Then such feasible values of $\Delta$, $L$, $G$, $X$ and \(\tilde{P}_\mathrm{com}\) will give a controller $K = L G^{-1}$ that (i) stabilizes the dynamics \eqref{eqn:global_dynamics} and (ii) has a sparsity structure such that 
    $E_{u,i} K E_{\xi,j}^\top = 0 \ \forall (j,i) \notin \mathcal{E}_C$, thus solving \(\mathcal{P}_\mathrm{SDS}\). 
\end{lem}
\begin{pf}
    If Assumptions \ref{assum:LTI_dynamics_unknown}, \ref{assum:Berberich_noise_bounds} and \ref{assum:combined_uncertainty_set} hold, any $\tilde{P}_\mathrm{com} \in \mathbf{\tilde{P}}_\mathrm{com}$, \(X \succ 0\) and $L$, $G$ satisfying \eqref{eqn:extended_stability_data_driven_LMI} will stabilize the dynamics \eqref{eqn:global_dynamics} by Thm. \ref{theorem:extended_data-driven_stability}. By \cite[Thm. 2]{Gross-2011-Optimized-Distributed-Control-Network-Topology-Design}, $K$ must have the desired sparsity pattern since \eqref{eqn:Gross-structure-constraints} is satisfied. \qed
\end{pf}
We combine the preceding results to formulate a data-driven control-aware communication topology design problem as an MISDP in order to solve \(\mathcal{P}_{\mathrm{Ctrl}\mhyphen\mathrm{Aw}}\) \eqref{eqn:ctrl_aw_framework_optim_formulation}. The cost function reflects a desire to prioritize communication between more strongly coupled subsystems. Therefore, if \(\delta_{ij} = 1\), we increase the cost by \(c_{ij}\) and reduce the cost by \(\eta_{ij}\). 
We arrive at the following optimization problem:
\begin{subequations}
\allowdisplaybreaks
\label{eqn:control_aware_comms_optim_MISDP}
	\begin{align}
    \label{eqn:cost_func}
		\min_{X \succ 0,\tilde{P}_\mathrm{com}, \Delta,L,G} & \sum_{i,j=1}^M \delta_{ij}(c_{ij} - \eta_{ij}) \\
        \label{eqn:Boolean_delta_ij_subeq}
        \mathrm{s.t.} \ \Delta = [\delta_{ij}]_{i,j=1:M}, \quad \delta_{ij}& \in \{0,1\} \ \forall i,j \in \mathcal{V}\\
        \eqref{eqn:extended_stability_data_driven_LMI}, \ \tilde{P}_\mathrm{com} &\in \mathbf{\tilde{P}}_\mathrm{com} \\
        \eqref{eqn:big_M_structured_set_constraints_1} - \eqref{eqn:big_M_structured_set_constraints_3}. &
	\end{align}
\end{subequations}
The characteristics of our control-aware topology optimization problem are described in Thm. \ref{theorem:MISDP_opt}:
\begin{thm}
\label{theorem:MISDP_opt}
    Consider the dynamics \eqref{eqn:extended_state_dynamics}, the control-aware topology design problem \(\mathcal{P}_{\mathrm{Ctrl}\mhyphen\mathrm{Aw}}\) and the optimization program \eqref{eqn:control_aware_comms_optim_MISDP}, which is convex in its continuous variables and is characterized as an MISDP. Suppose that Assumptions \ref{assum:LTI_dynamics_unknown}, \ref{assum:Berberich_noise_bounds} and \ref{assum:combined_uncertainty_set} hold. Then a feasible solution provides a cyber-layer graph $\mathcal{G}_C$ such that: \\
    (i) The controller $K = L G^{-1}$ respects the sparsity pattern specified by $\mathcal{G}_C$; that is, $E_{u,i} K E_{\xi,j}^\top = 0 \ \forall (j,i) \notin \mathcal{E}_C$, \\
    (ii) $K$ stabilizes the extended state dynamics \eqref{eqn:extended_state_dynamics}.\\
    An optimal solution solves \(\mathcal{P}_{\mathrm{Ctrl}\mhyphen\mathrm{Aw}}\) \eqref{eqn:ctrl_aw_framework_optim_formulation} and will have an associated cyber-layer edge set $\mathcal{E}_C$ that gives the minimum sum of the coupling estimates $\eta_{ij}$ subtracted from the included communication link costs; that is, $\sum_{(j,i) \in \mathcal{E}_C} c_{ij} - \eta_{ij}$, whilst satisfying feasibility.
\end{thm}
\begin{pf}
    For any feasible solution satisfying the constraints, (i) is proven by Lemma \ref{lemma:gross_stabilizing_structure} since there is an edge \((j,i) \in \mathcal{E}_C\) if and only if \(\delta_{ij} = 1\). Under Assumptions \ref{assum:LTI_dynamics_unknown}, \ref{assum:Berberich_noise_bounds} and \ref{assum:combined_uncertainty_set}, stability (ii) is proven as a consequence of Thm. \ref{theorem:extended_data-driven_stability}. If $(j,i) \in \mathcal{E}_C$ then $(c_{ij} - \eta_{ij})$ is included in the optimal cost due to \eqref{eqn:cost_func}, where $\mathcal{E}_C$ is determined by the optimal solution $\Delta^*$ and \(c_{ij}\) represents a link cost, completing the proof.
    \qed
\end{pf}
The formulation \eqref{eqn:control_aware_comms_optim_MISDP} satisfies our conditions for a control-aware communication topology design. We have \(\mathcal{C}_\mathrm{stab}\) given by \eqref{eqn:extended_stability_data_driven_LMI} whilst \(\mathcal{C}_\mathrm{str}\) is encoded by \eqref{eqn:big_M_structured_set_constraints_1}-\eqref{eqn:big_M_structured_set_constraints_3} (cf. \eqref{eqn:ctrl_aw_framework_optim_formulation}). 
\begin{rem}
    Removing \(\eta\) from \eqref{eqn:cost_func} minimizes communication cost under the stated synthesis constraints. If \(c_{ij}>0\) for \(i\ne j\) and \(c_{ii}=0\), a zero optimum certifies zero unstable design-dependent DFMs and feasibility of decentralized stabilization under these constraints.
\end{rem}

\begin{rem}
    A worst-case performance requirement can be imposed by adding a \(\mathcal{C}_\mathrm{perf}\), for example \eqref{eqn:extended_H2_performance_LMI_thm} and \(\operatorname{trace} \, (\Gamma) \leq \gamma_\mathrm{max}^2\), to \eqref{eqn:control_aware_comms_optim_MISDP}. This interpolates between purely control-aware design and full co-design.
\end{rem}

\section{Numerical Examples}
\label{sec:simulations}

Datasets are generated by exciting each system with \(u_i^d(k)\overset{\mathrm{i.i.d.}}{\sim}\mathcal{U} [-1,1]\) and recording the resulting outputs. Communication topologies are designed using the proposed control-aware and co-design methods. To compare the quality of the selected topologies consistently, a structured \(\mathcal{H}_2\) controller is subsequently synthesized for each topology by minimizing \(\operatorname{trace} \, (\Gamma)\) subject to the proposed stability \eqref{eqn:extended_stability_data_driven_LMI}, structural \eqref{eqn:Gross-structure-constraints} and performance \eqref{eqn:extended_H2_performance_LMI} constraints. We define \(\mathbf{\tilde{P}}_\mathrm{com}\) using prior knowledge in the form of norm bounds on the true uncertainty \(\Psi_\mathrm{tr} \Psi_\mathrm{tr}^\top \preceq \bar{\psi} I\) and disturbance \(\tilde{D}_\mathrm{tr} \tilde{D}_\mathrm{tr}^\top \preceq \bar{d} I\); see \cite{Berberich-et-al-2023-Prior-Knowledge-Data-Robust-Design} Sect. V.B. for details. We let \(c_{ij} = c \ \forall i,j \neq i \), \(c_{ii} = 0 \ \forall i \in \mathcal{V}\), and we vary \(c\) over an appropriate grid to produce distinct optimized link counts. All optimization programs are formulated in MATLAB using the YALMIP toolbox \cite{Lofberg-2004-YALMIP} and solved using YALMIP BNB and Mosek \cite{mosek}, with computations performed on an Intel 11\textsuperscript{th} Gen i7 processor with 16 GB RAM.

\subsection{4-Subsystem Chain Example}
\label{subsec:4-subsys_sims}



Consider the discretized and linearized swing equations commonly used to model small-signal frequency dynamics in electric power transmission systems as given in the example in \cite{Alonso-2022-Data-Driven-Distributed-Localized-MPC}. Each subsystem has the dynamics \eqref{eqn:subsys_dynamics}, with {\renewcommand{\arraystretch}{1}$x_i(k) = \begin{bmatrix}
    \theta_i(k) & \omega_i(k)
\end{bmatrix}^\top$, $B_i = F_i = \begin{bmatrix}
    0 & \frac{1}{h_i} \Delta t
\end{bmatrix}^\top$, $C_i = \begin{bmatrix}
        1 & 0
    \end{bmatrix}$, $D_i = \mathbf{0}^{p_i \times m_i}$ and:
\begin{equation}
\label{eqn:sim_system}
    A_{ii} = \begin{bmatrix}
        1 & \Delta t \\
        -\frac{k_i}{h_i} \Delta t & 1 - \frac{d_i}{h_i} \Delta t
    \end{bmatrix}, \ A_{ij} = {\begin{bmatrix}
        0 & 0 \\
        \frac{k_{ij}}{h_i} \Delta t & 0
    \end{bmatrix}} \ \forall j \neq i,
\end{equation}}

where $\theta_i$ and $\omega_i$ represent phase angle and frequency deviations, and $h_i$, $d_i$ and $k_{ij}$ represent inertia, damping and coupling for each subsystem, respectively, with the time step $\Delta t = 0.2 \ \mathrm{s}$ and $k_i = \sum_{j \in \mathcal{V}\backslash \{i\}}{k_{ij}}$.
We simulate a chain of \(M=4\) subsystems 
with parameters: 
\(h = \left[1.4 \ 0.5 \ 0.6 \ 1 \right]^\top\), \(d = \left[ 1.1 \ 0.8 \ 0.9 \ 1.4 \right]^\top\), \(k_{12} = k_{21} = 1.25\), \(k_{23} = k_{32} = 2.25\), \(k_{34} = k_{43} = 0.5\), and all other \(k_{ij} = 0\). We have \(p = m = 4\), \(\ell = 2\) and \(n = 8\).
We choose {\renewcommand{\arraystretch}{1}\(C_e = \begin{bmatrix}
    \mathrm{diag}(0, I) & 0
\end{bmatrix}^\top\), \(D_{e} = \begin{bmatrix}
    0 & I
\end{bmatrix}^\top\)}, and \(\bar{M} = 1000\). We design structured \(\mathcal{H}_2\) controllers for all 64 topologies with bidirectional links to assess the optimality of our topology design.
For noise-free simulations, we take \(T = 100\) and \(\bar{\psi} = 20\), whilst with noise, we choose \(T = 300\), \(\bar{\psi} = 5\) and model the process disturbance and measurement noise as \(\varepsilon(k) \overset{\mathrm{i.i.d.}}{\sim} \mathcal{U}[-0.02,0.02]\) and \(v(k) \overset{\mathrm{i.i.d.}}{\sim} \mathcal{U}[-0.001,0.001]\), with a corresponding \(\bar{d} = 0.05\).
An upper bound on control performance is obtained during synthesis (i.e., the optimal value of \(\operatorname{trace} \, (\Gamma)\)), whilst a realized cost is found by closed-loop system performance analysis following controller design.
Relative to the exhaustive bidirectional comparison set, co-design attains the lowest synthesis bounds and realized closed-loop cost, whilst control-aware design attains the lowest or near-lowest values of both (Fig. \ref{fig:4-subsystem_results}). The synthesis bound is conservative, especially for sparse topologies (Fig. \ref{fig:4-subsystem_results}(b)); this gap is consistent with conservatism introduced by the sufficient structural conditions \eqref{eqn:Gross-structure-constraints}. 

Performance under control-aware topologies is degraded when the number of links increases from \(3\) to \(6\). This non-monotonicity is consistent with conservatism of \eqref{eqn:Gross-structure-constraints}: when a link is added to a given topology the constraints on \(L\) and \(G\) are changed, such that a feasible solution may not remain feasible for the new topology. To investigate this effect, we find topologies using a model-based sparsity-promoting control algorithm \cite{Lin-2013-Design-Optimal-Sparse-Feedback-Gains-via-ADMM}, then design data-driven controllers for each model-based topology. As shown in Fig. \ref{fig:4-subsystem_results}(a), control performance under the model-based topologies displays the same non-monotonicity as the data-driven control-aware topology results.
\begin{figure*}[h!]
    \centering
    \includegraphics[width=0.95\textwidth]{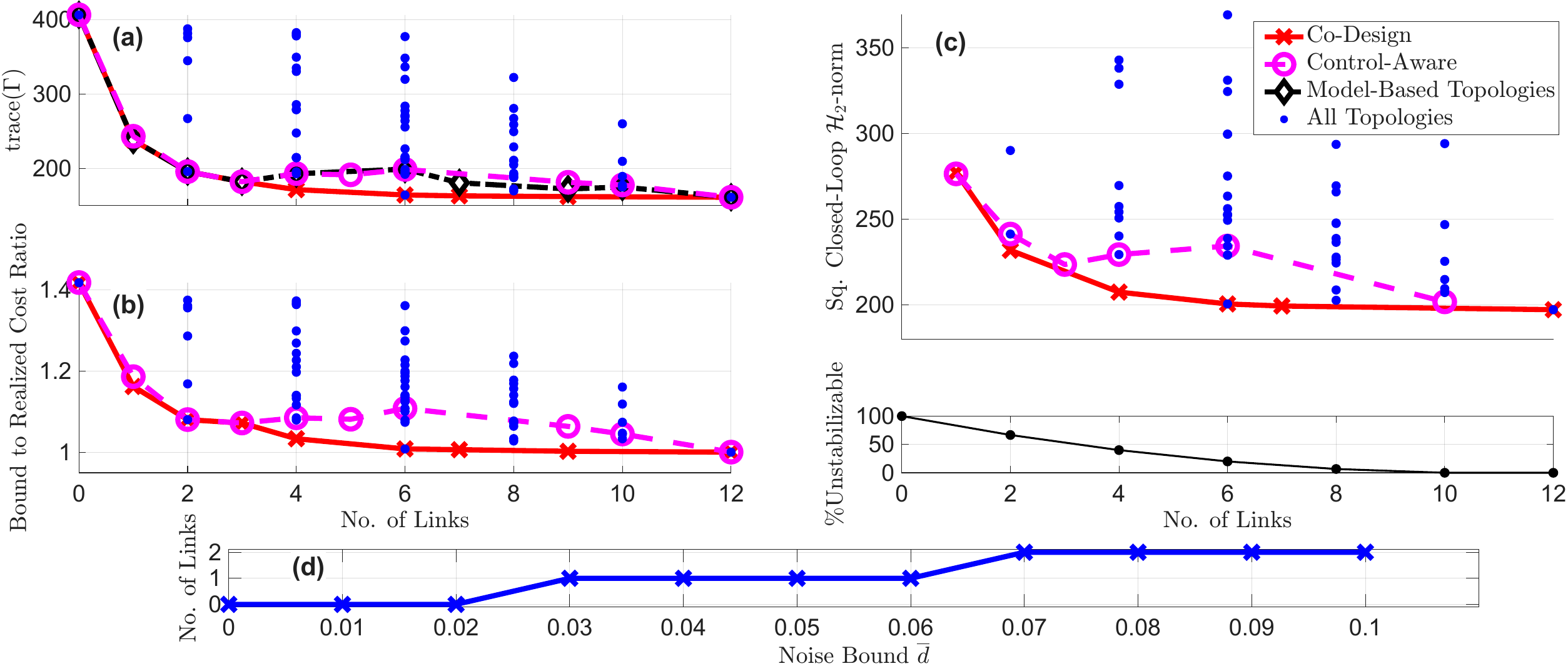}
    \caption{Results for 4-subsystem experiments. (a) Noise-free synthesis cost bound. (b) Ratio of the noise-free synthesis bound to the realized squared \(\mathcal{H}_2\) cost. (c) With noise (\(\bar{d} = 0.05\)), showing (upper panel) realized squared \(\mathcal{H}_2\) cost and (lower panel) percentage of infeasible bidirectional topologies. (d) Minimum feasible communication link count against disturbance bound \(\bar{d}\).}
    \label{fig:4-subsystem_results}
\end{figure*}
Results in Fig. \ref{fig:4-subsystem_results}(d) show that as the noise bound increases from zero, the number of links required to stabilize the system under our synthesis conditions increases from zero to one then two, before our robust stabilization condition becomes infeasible. Standard fixed-mode analysis alone cannot reveal this loss of robust feasibility since the decentralized topology is feasible in the noise-free case. Control-aware design achieved appreciably lower computation times than co-design across a range of communication costs for chains of 4, 7, 9 and 14 subsystems, as shown in Fig. \ref{fig:computation_time}.
\begin{figure}[h]
    \centering
    \includegraphics[width=0.6\linewidth]{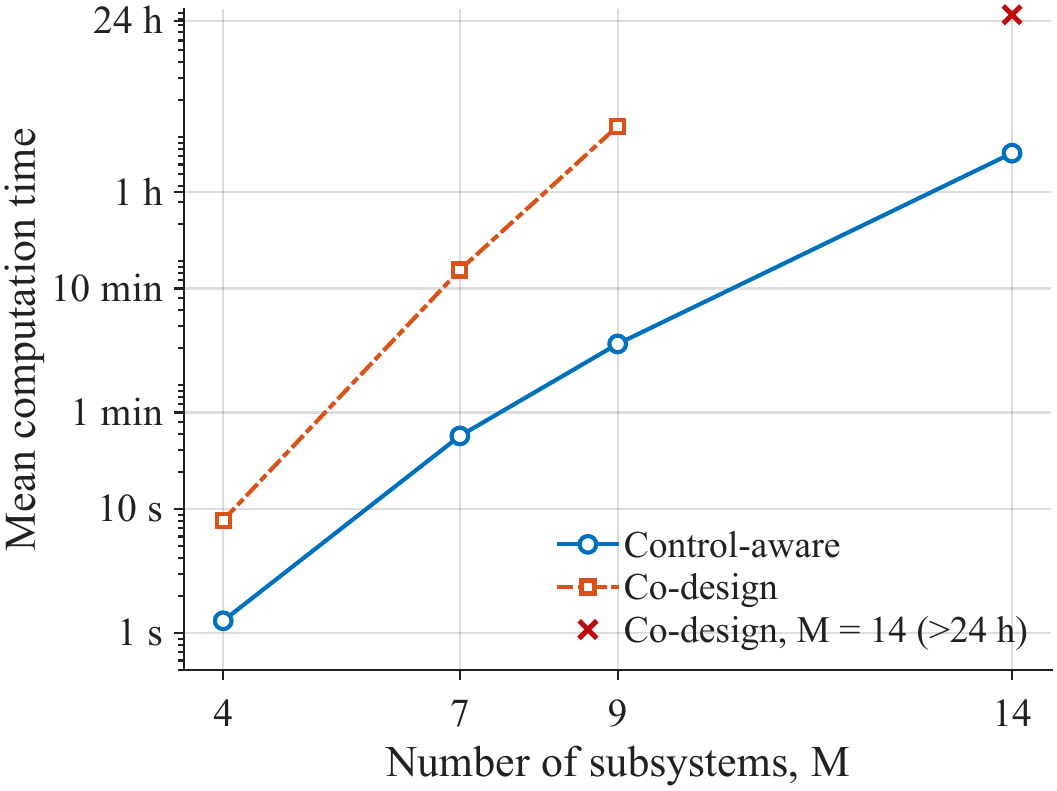}
    \caption{Mean computation times of control-aware design and co-design for chain systems, with a 24 hour timeout.}
    \label{fig:computation_time}
\end{figure}

\subsection{Decentralized Fixed Mode Example}
\label{subsec:DFM_sims}

We consider the three-agent system in \cite{Gross-2011-Optimized-Distributed-Control-Network-Topology-Design}, where \(x_i(k) \in \mathbb{R}^2\), \(u_i(k) \in \mathbb{R}\) \(\forall i \in \{1,2,3\}\), there is no process disturbance (\(\bar{d} = 0\)) and dynamics are LTI. We assume full-state measurements (\(y(k) = x(k)\)) so \(\ell = 1\). The discrete-time system has a DFM on the unit circle, which prevents decentralized asymptotic stabilization. However, it can be achieved if Agent 1 receives state measurements from Agent 2.
Our previous work \cite{Nestor-2025-SysDO} is unable to guarantee asymptotic stabilization of this system as it assumes that any DFMs are stable. Our proposed synthesis conditions here only identify communication topologies admitting a stabilizing controller. We choose \(T = 100\), \(\bar{\psi} = 20\), \(\bar{M} = 1000\), and the performance signal such that it is equivalent to the cost weights \(Q\), \(R\) in \cite{Gross-2011-Optimized-Distributed-Control-Network-Topology-Design}, map the scheme in \cite{Gross-2011-Optimized-Distributed-Control-Network-Topology-Design} to the extended state system and compare the results to our data-driven approaches in Table \ref{tab:DFM_results_Gamma}. Data-driven co-design closely matches the model-based co-design benchmark, whilst the control-aware approach is conservative for the 4-link topology as discussed above in Sect. \ref{subsec:4-subsys_sims}.
\begin{table}[h]
   \centering
   \caption{Control cost upper bound, \(\operatorname{trace} \, (\Gamma)\), for Sect. \ref{subsec:DFM_sims} example.}
   {\renewcommand{\arraystretch}{1.4}
   \begin{tabular}{c|c|c|c}
\hline
   \# Links & Control-aware & Co-design & Model \cite{Gross-2011-Optimized-Distributed-Control-Network-Topology-Design} \\ \hline \hline
   6 & 76.4 & 76.4 & 76.4  \\ \hline
   4 & 120.8 & 87.5 & 87.2 \\ \hline
   2 & 94.5 & 94.5 & 94.5 \\ \hline
   1 & 125.0 & 125.0 & 125.0 \\ \hline
   \end{tabular}}
   \label{tab:DFM_results_Gamma}
\end{table}

\subsection{IEEE 14-Bus Example}
\label{subsec:14-subsys_sims}


\begin{figure}
    \centering
    \includegraphics[width=0.75\linewidth]{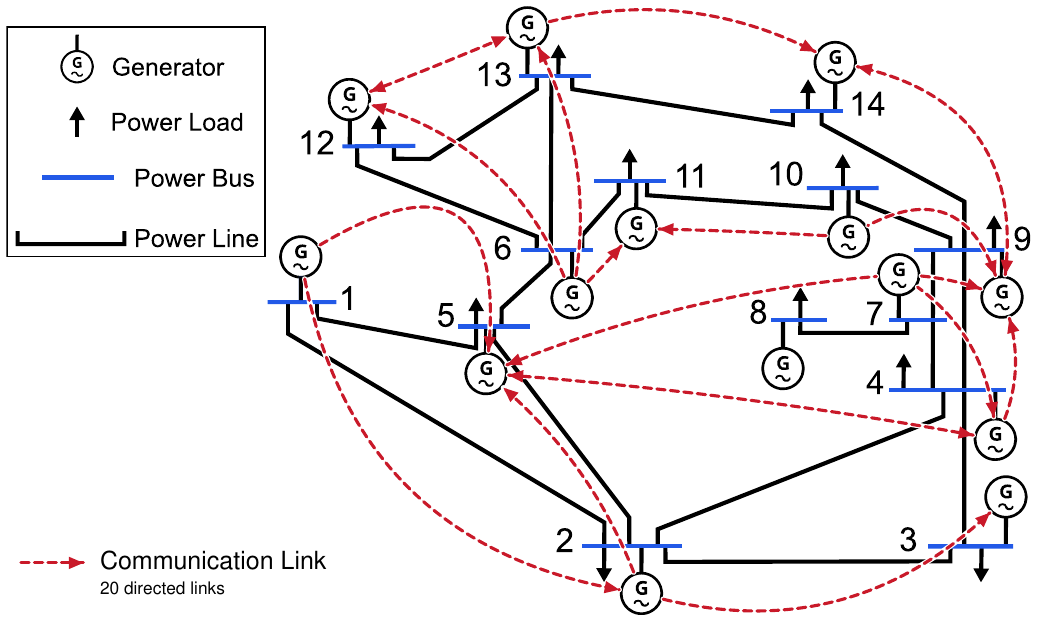}
    \caption{The IEEE 14-bus test system topology, with a control-aware 20 link communication topology.}
    \label{fig:IEEE-14-bus}
\end{figure}
We consider the IEEE 14-bus test system for electric power transmission systems and assume inertia and damping are present at all nodes. The subsystem dynamics are given by \eqref{eqn:subsys_dynamics}, with the definitions of \(x_i(k)\), \(A_{ij}\), \(B_i\), \(C_i\) and \(D_i\) the same as in Sect. \ref{subsec:4-subsys_sims}. The physical coupling topology of the test system is shown in Fig. \ref{fig:IEEE-14-bus}, with parameters drawn from uniform distributions according to \(h_i \sim \mathcal{U}[0.01,2]\), \(d_i \sim \mathcal{U}[0.5,1]\), \(k_{ij} \sim \mathcal{U}[1,1.5]\), \(k_{ij} = k_{ji} \ \forall i,j\) if \(i\) and \(j\) are physically coupled; otherwise \(k_{ij} = k_{ji} = 0\). We use \(\Delta t = 0.2\) s, \(T = 850\), \(\bar{\psi} = 100\) and \(\bar{M} = 1000\). We optimize the communication topology using the control-aware approach for a range of communication costs without any noise or disturbance, and compute a structured \(\mathcal{H}_2\) controller following topology design.
The achieved control cost and its synthesis upper bound are closely correlated, with cost and topology sparsity increasing with communication cost (Fig. \ref{fig:14_Bus_Results}). No improvements in control performance are obtained once physically coupled subsystems communicate. Control-aware design took around one hour on average, whilst co-design did not terminate within a 24 hour time limit. 
\begin{figure}[h]
    \centering
    \includegraphics[width=0.95\linewidth]{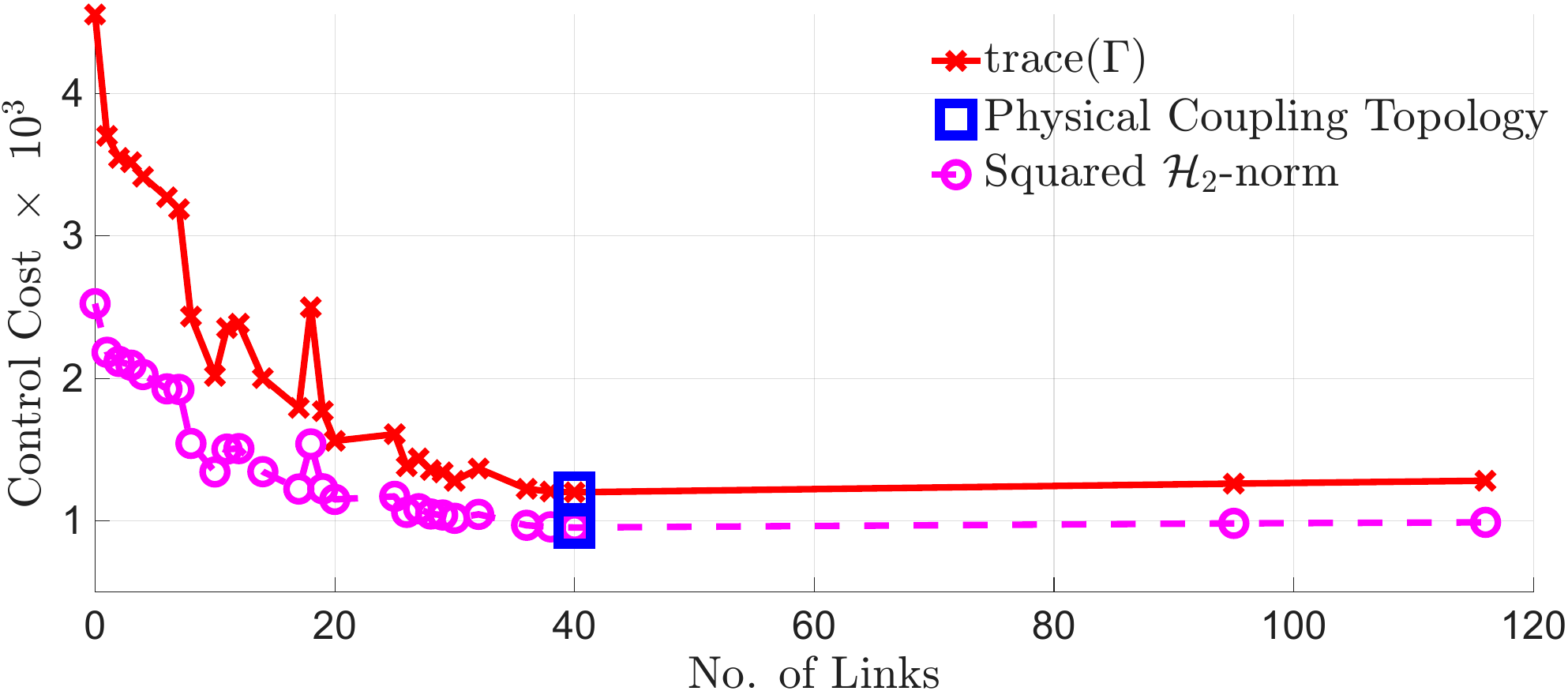}
    \caption{Closed-loop squared \(\mathcal{H}_2\)-norm (magenta) and control cost upper bound, \(\operatorname{trace} \, (\Gamma)\) (red), plotted against communication link count for the IEEE 14-bus test case.}
    \label{fig:14_Bus_Results}
\end{figure}

\section{Conclusion}
\label{sec:conclusion}


In this paper, we proposed a control-aware approach to communication topology design, coupling control structure characterization with convex synthesis methods. A data-driven scheme for control-aware design was presented as an MISDP, including a control performance proxy based on the open-loop coupling strength between subsystems. A data-driven co-design MISDP was developed, minimizing a certified upper bound on the squared closed-loop \(\mathcal{H}_2\)-norm together with communication costs. Simulations demonstrated that the communication required for data-driven stabilization feasibility depends on dataset noise, and in the tested examples, the control-aware method produced competitive topologies with lower computation times than co-design. However, the condition \(p \ell = n\) is necessary for our approach, limiting its real-world applicability. Future avenues for investigation include aiming to relax this condition and considering topology design in the distributed synthesis setting.


\ack

The authors would like to thank Dr Yulong Gao, Dr Boli Chen, Dr Julian Berberich, and Dr Giordano Scarciotti for insightful discussions and advice.


\appendix

\bibliographystyle{ieeetr}
\bibliography{bibliography}

\begin{wrapfigure}[15]{L}{0.14\textwidth}
    \centering
    \includegraphics[width=0.16\textwidth]{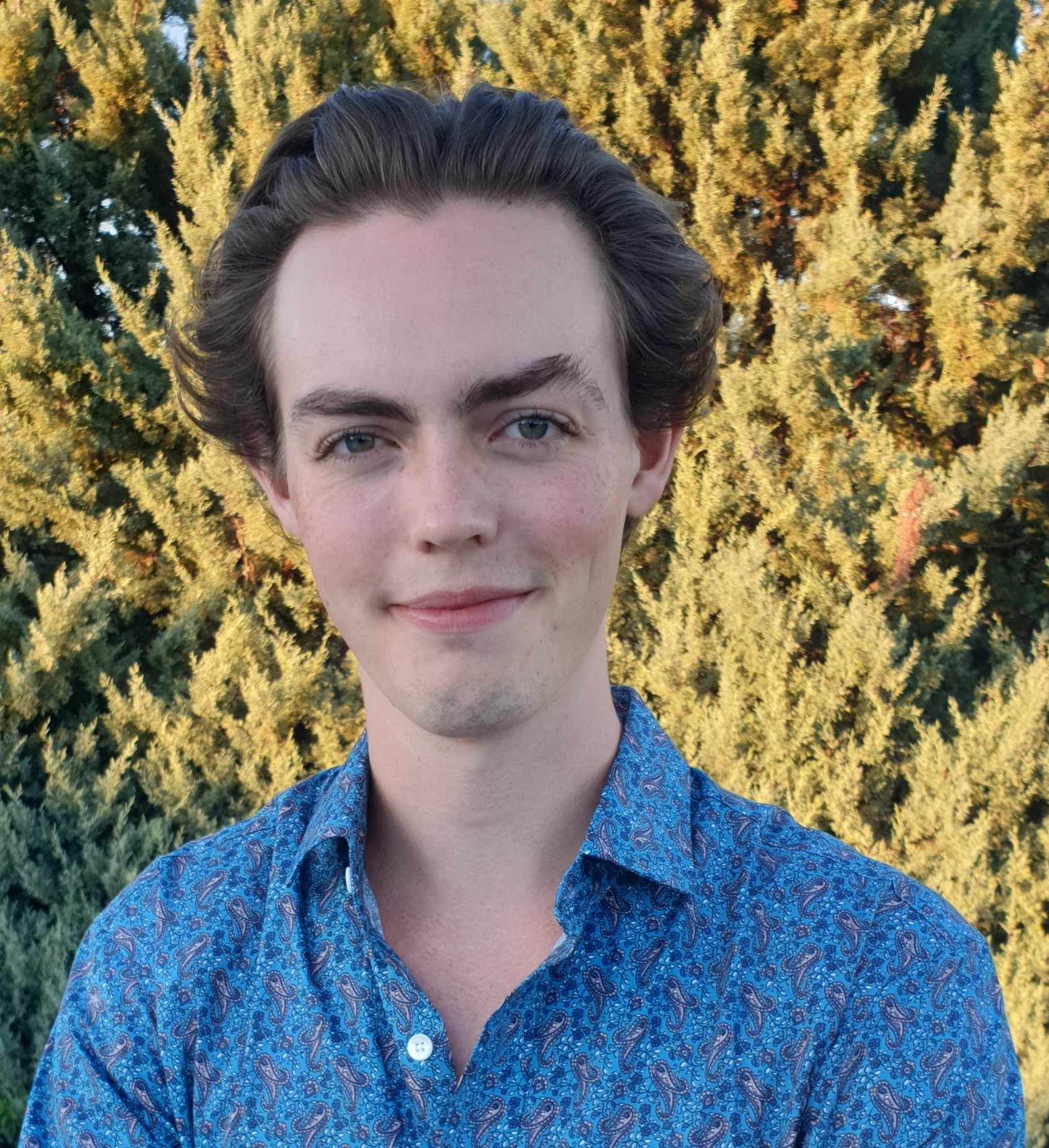}
\end{wrapfigure} 
\textbf{Michael Nestor} received the M.Eng. degree in engineering science in 2021 from the University of Oxford, U.K. He is currently working towards the Ph.D. degree at Imperial College, London, U.K. He was a visiting student at Tsinghua University in 2025. His research interests include distributed data-driven control of interconnected systems, trade-offs in communication and control system design, and power system service provision using distributed energy resources.

\begin{wrapfigure}[13]{l}{0.14\textwidth}
    \centering
    \includegraphics[width=0.16\textwidth]{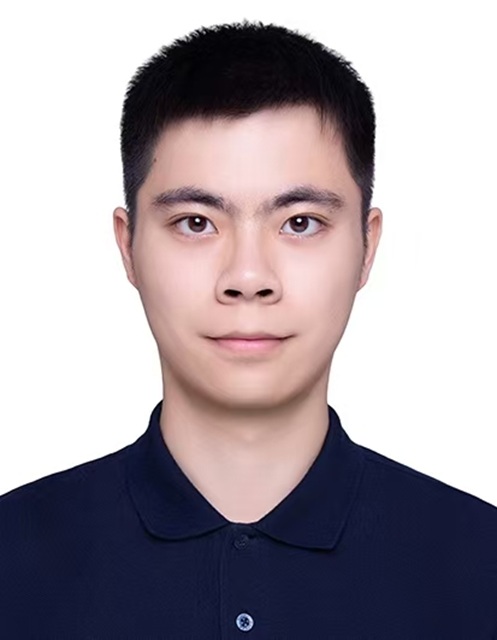}
\end{wrapfigure}
\textbf{Jiaxin Wang} received the B.S. degree in electrical engineering in 2022 from Tsinghua University, Beijing, China, where he is currently working towards the Ph.D. degree. He was a visiting student at Imperial College, London in 2025/2026. His research interests include power system scheduling and planning, stability analysis and stability constrained optimization and control, and AI techniques in power systems.

\begin{wrapfigure}[17]{l}{0.14\textwidth}
    \centering \includegraphics[width=0.16\textwidth]{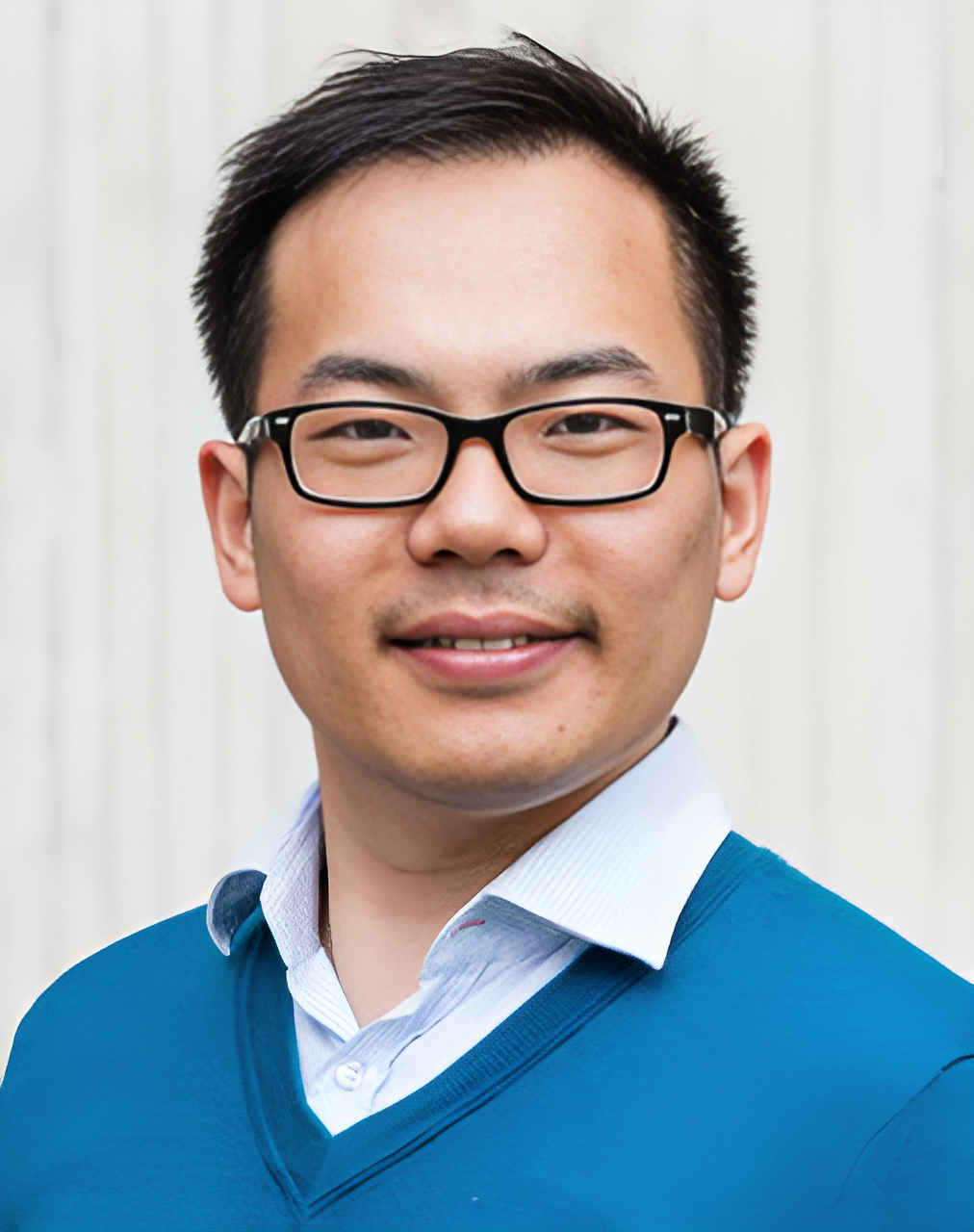}
\end{wrapfigure}
\textbf{Fei Teng} received the B.Eng. degree in electrical engineering from Beihang University, China, in 2009, and the M.Sc. and Ph.D. degrees in electrical engineering from Imperial College, London, U.K., in 2010 and 2015, respectively, where he is currently a Reader in Intelligent Energy Systems with the Department of Electrical and Electronic Engineering. His research focuses on power system operation with a high penetration of inverter-based resources, and the cyber-resilient and privacy-preserving cyber-physical power grid.

\end{document}